\documentclass{statsoc}
\usepackage{amsthm}
\usepackage[a4paper]{geometry}
\usepackage{float}
\usepackage{graphicx}
\usepackage[table,xcdraw]{xcolor}
\usepackage{breakcites}
\usepackage{natbib}
\usepackage{xurl}
\usepackage{subcaption}
\usepackage{longtable}
\usepackage{multirow,verbatim}
\usepackage{epsfig}
\usepackage{graphicx}
\usepackage{amsmath}
\usepackage{amsfonts}
\usepackage{amssymb}
\usepackage{booktabs}

\RequirePackage[colorlinks,citecolor=blue,urlcolor=blue]{hyperref}

\newcommand{\Appendix}
{
\def\thesection{Appendix~\Alph{section}}
\def\thesubsection{A.\arabic{subsection}}
}

\theoremstyle{plain}
\numberwithin{equation}{section} \theoremstyle{plain}
\newtheorem{thm}{Theorem}[section]

\newtheorem*{repthm*}{Theorem 1}
\newtheorem{cor}{Corollary}
\newtheorem{lem}{Lemma}

\theoremstyle{definition}
\newtheorem{defn}{Definition}[section]
\theoremstyle{remark}

\title[Sliced Elastic Distance]{Sliced Elastic Distance for Evaluating Amplitude and Phase Differences in Precipitation Models}

\author{Robert C. Garrett}
\address{Department of Statistics, 
University of Illinois Urbana-Champaign,
Champaign, IL,
US.}

\author{Trevor Harris}
\address{Department of Statistics,
University of Connecticut
Storrs, CT,
US.}

\author{Zhou Wang}
\address{Department of Climate, Meteorology \& Atmospheric Sciences,
University of Illinois Urbana-Champaign,
Champaign, IL,
US.}

\author[Garrett {\it et al.}]{Bo Li}
\address{Department of Statistics and Data Science,
Washington University in St. Louis,
St. Louis, MO,
US.}

\begin{document}
\begin{abstract}
Climate model evaluation plays a crucial role in ensuring the accuracy of climatological predictions. However, existing statistical evaluation methods often overlook time misalignment of events in a system's evolution, which can lead to a failure in identifying specific model deficiencies. This issue is particularly relevant for climate variables that involve time-sensitive events such as the monsoon season. To more comprehensively evaluate climate fields, we introduce a new vector-valued metric, the sliced elastic distance, through kernel convolution-derived slices. This metric simultaneously and separately accounts for spatial and temporal variability by decomposing the total distance between model simulations and observational data into three components: amplitude differences, timing variability, and bias (translation). We use the sliced elastic distance to assess CMIP6 precipitation simulations against observational data, evaluating amplitude and phase distances at both global and regional scales. In addition, we conduct a detailed phase analysis of the Indian Summer Monsoon to quantify timing biases in the onset and retreat of the monsoon season across the CMIP6 models.
\end{abstract}

\newpage
\section{Introduction}
Climate models are essential tools for studying changes in the climate system in response to anthropogenic forcings \citep{kattenberg1996climate}. Scientists have continually refined these models to accurately represent the real dynamics of the climate systems. Many of the models introduced in the Coupled Model Intercomparison Project Phase 6 (CMIP6) represent significant advancements in simulation complexity and resolution \citep{eyring2016overview}. Model evaluation is a critical aspect of this improvement process, providing users with valuable insights into the model's strengths and limitations \citep{ye2023evaluation}. Consequently, developing comprehensive evaluation metrics is crucial for assessing the model's ability to capture the complex characteristics of climate \citep{randall2007climate, eyring2019taking}.

Climate model evaluation involves comparing model outputs with observational datasets \citep{flato2014evaluation}. This can be achieved through various approaches, including similarity measures, spatial data modeling, and functional data analysis (FDA) techniques. Widely used similarity measures include root mean square error (RMSE), mean absolute error (MAE), Taylor skill score, inter-annual variability score, and correlation coefficients \citep{li2022evaluation, yazdandoost2021evaluation, du2022comprehensive, ngoma2021evaluation}. Development in spatial statistics includes hypothesis testing for spatial data \citep{shen2002nonparametric, cressie2008detecting, yun2022detection}, quantifying loss differentials \citep{snell2000spatial, wang2007intraseasonal, hering2011comparing}, and examining the first- and second-order dependency structures of spatial processes \citep{lund2009revisiting, li2012defining}. 
FDA based approaches have been used to compare means \citep{ramsay2005fitting, zhang2007statistical, horvath2013estimation, staicu2014likelihood}, covariance structures \citep{zhang2015two, li2016comparison}, and distributional differences \citep{harris2021evaluating} by treating spatiotemporal processes as discrete representations of continuous functions.

Precipitation is one of the most challenging variables for climate models to represent due to its complex interaction with the climate system and the limitations of observational data products \citep{nelson2016assessment}. It is also one of the most important climate variables  \citep{trenberth2011changes} for its strong influence on agriculture, energy production \citep{ramachandra2007spatial} and many other economic sectors \citep{wang2006asian, prasanna2014impact}.
In particular, large-scale seasonal shifts in precipitation, such as those associated with monsoons, can severely disrupt agriculture and industry if they occur earlier or later than predicted by climate models \citep{ye2023origin}.
Despite this importance, most current evaluations of monsoonal precipitation in climate models have primarily focused on mean precipitation biases during the monsoon season \citep{katzenberger2021robust, konda2023evaluation, xin2020comparison}. More recently, some studies have begun to assess model biases in the timing of monsoon onset and retreat \citep{ye2023origin, khadka2022evaluation, ha2020future}. Motivated by these efforts, our work seeks to rigorously quantify timing biases in monsoonal precipitation as represented in climate models.

Although climate model validation methods have advanced in recent years, no existing approach fully disentangles the timing of events from their magnitude, which can result in imprecise estimates of timing, intensity, and variability \citep{tucker2013generative, srivastava2016functional}.
Therefore, we introduce a climate model evaluation metric, sliced elastic distance, which fully isolates the two sources of variability to comprehensively and rigorously assess model performance.
Our approach extends the elastic shape metrics \citep{srivastava2011registration} to spatiotemporal processes on a sphere via spherical convolutional slicing \citep{garrett2024validating}.
Convolutional slicing allows to locally align spatiotemporal fields continuously over the globe which, we show, leads to a proper vector-valued metric \citep{sastry2012common}. Furthermore, local alignments lead to local distance maps \citep{garrett2024validating} that can be used to study spatial variations in phase-amplitude variability. We apply our method to rank CMIP6 daily precipitation simulations using observational data as a reference, and then map the local onset and retreat biases of the Indian Summer Monsoon in each model. 

The remainder of the article is organized as follows. Section \ref{sec:data} describes the various precipitation datasets used in our analysis. Section \ref{sec:method} proposes the sliced elastic distance and establishes its theoretical properties. Section \ref{sec:sim} demonstrates the performance of the sliced elastic distance through simulated model validation scenarios. Section \ref{sec:results} applies our distance metric to evaluate CMIP6 precipitation models in terms of  amplitude and phase differences on both global and local scales. Finally, Section \ref{sec:discussion} provides a brief discussion of the method and the results.

\section{Data Description}
\label{sec:data}

We obtain precipitation outputs from the CMIP6 historical experiments \citep{ESGF_LLNL}, to assess the global and local performance of climate models. Daily precipitation data (mm) was chosen to provide the necessary temporal resolution for accurate assessment of monsoon timing, consistent with the approach of \cite{misra2018local}.
We use an ensemble of 45 model outputs (Appendix \ref{a:data1}), each run under the \texttt{r1i1p1f1} variant ID, representing a unique simulation of historical precipitation fields from a different model using the same initialization, physics, and forcings \citep{eyring2016overview}. 

We use the National Centers for Environmental Information (NCEI) Global Precipitation Climatology Project (GPCP) Daily Precipitation Analysis Climate Data Record \citep{huffman2001global,adler2020global} as the reference for climate model evaluation. Data were obtained from the GPCP Daily V1.3 analysis publicly available on the NCEI website \citep{NCEI_Precipitation_GPCP}.
We also include precipitation fields from the European Centre for Medium-Range Weather Forecasts (ECMWF) Reanalysis 5th Generation (ERA5) \citep{hersbach2020era5, hersbach_era5_2023} and the National Centers for Environmental Prediction (NCEP) Reanalysis-2 dataset \citep{kanamitsu2002ncep} as points of comparison given their well known biases \citep{nelson2016assessment, zhou2022development}. The ERA5 hourly data \citep{Copernicus_ERA5} were aggregated to a daily frequency to match the NCEP Reanalysis-2 \citep{NOAA_NCEP_Reanalysis2} data.

Each data product provides precipitation values on a regular latitude-longitude grid with a variable size and structure, with the exception of the ICON-ESM-LR climate model. We evaluate each of the 45 CMIP6 models over the historical period from January 1997 to December 2014, a common period covered by GPCP, ERA5, NCEP, and the models.

\section{Methods}
\label{sec:method}

\subsection{Review of Elastic functional data analysis}
\label{sec:efda}

Elastic functional data analysis (EFDA) \citep{joshi2007novel, tucker2013generative, srivastava2016functional} is a general framework for comparing the shapes of absolutely continuous manifold-valued functions parameterized by time.
Because our method will warp univariate projections of spatiotemporal processes, we introduce the EFDA framework for special case of univariate functional data.

Let $\mathcal{F_T}$ denote the space of absolutely continuous functions from $\mathcal{T} \mapsto \mathbb{R}$, where $\mathcal{T} = [0, 1]$ without loss of generality and let $f, g \in \mathcal{F_T}$ denote two functions. To compare the intrinsic shape differences between $f$ and $g$, EFDA introduces a square root velocity function (SRVF) representation and warps the SRVF of $f$ to the SRVF of $g$ \citep{joshi2007novel}. The SRVF of a function $f \in \mathcal{F_T}$ is defined as $q_{f(t)} = \text{sign}(\dot{f}(t))\sqrt{|\dot{f}(t)|}$ \citep{srivastava2016functional}, where $\dot{f}(t)$ denotes the time derivative of $f(t)$. The SRVF provides a unique and invertible representation of $f$, up to vertical translation.
Warping is accomplished by estimating warping functions $\gamma_f, \gamma_g \in \Gamma$  that minimize $||q_{f(\gamma_f(t))} - q_{g(\gamma_g(t))}||_2$. The space $\Gamma$ includes all boundary-preserving, absolutely continuous, and weakly increasing functions on $\mathcal{T}$.

Following \cite{srivastava2016functional}, we the define amplitude, phase, and translation distances, denoted $D_A$, $D_P$, and $D_T$, respectively, between $f$ and $g$ as:
\begin{equation}\label{eq:3distances}
\begin{aligned}
  D_A(f,g) &= \text{inf}_{\gamma_f,\gamma_g \in \Gamma}||(q_f,\gamma_f) - (q_g,\gamma_g)||_2, \\
  D_P(f,g) &= \text{cos}^{-1}\left(\int_0^1\sqrt{\dot\gamma_f^*(t)}\sqrt{\dot\gamma_g^*(t)}dt\right), \\
  D_T(f,g) &= |f(0)-g(0)|,
\end{aligned}
\end{equation}
where 
$(q_f,\gamma) = q_{f(\gamma(t))} = (q_f(\gamma(t)))\sqrt{\dot\gamma}$ is the SRVF of $f(\gamma(t))$, the time warping of $f$ by $\gamma$, and $\dot\gamma$ is the time derivative of $\gamma$. $D_P(f,g)$ quantifies the degree of time warping required to align $f$ and $g$ in time, and, therefore, represents timing differences between $f$ and $g$ regardless of their individual magnitudes. $D_A(f,g)$ quantifies the residual magnitude difference between $f$ and $g$ after aligning them in time. Finally, $D_T(f,g)$ quantifies the bias between $f$ and $g$, irrespective of their timing or magnitude.

The distance functions $D_A$, $D_P$, and $D_T$ derived from the SRVF representation are all proper metrics on their respective spaces \citep{srivastava2016functional}. 
The triplet of amplitude, phase, and translation distance together is called the elastic distance, denoted as:
\begin{equation}
  \label{d:elastic}
  D_E(f,g) = \Bigl[D_A(f,g), D_P(f,g), D_T(f,g)\Bigr]^T,
\end{equation}
where $x^T$ is the transpose of a row vector $x$. EFDA aims to estimate the elastic distance between functions, which fully characterizes their differences, while allowing for meaningful separation of variability. We prove a result of independent interest that $D_E$ is a vector-valued metric \citep{sastry2012common} on $\mathcal{F_T}$ in Section \ref{sec:sedtheory}.

\subsection{Sliced elastic distance}
\label{sec:sed}

To develop our approach, we re-cast spatiotemporal climate fields as continuous functions of space-time. Let $\mathcal{S} = \mathbb{S}^2$, the unit sphere, denote the spatial domain and $\mathcal{T} = [0, 1]$ denote the time domain. Each spatiotemporal field is a continuous function $f(s,t) \in \mathcal{F_{\mathcal S\times T}}$, where $\mathcal{F_{\mathcal S\times T}}$ is the set of continuous functions $f: \mathcal S\times \mathcal T \mapsto \mathbb{R}$, such that, for any fixed location $s_0$, the component function $f(s_0,t)$ is also an absolutely continuous function of time. 

Drawing from the sliced Wasserstein distance \citep{garrett2024validating} for functional data, we introduce a sliced \textit{elastic} distance. Slices are random \citep{cuevas2007robust} or deterministic \citep{delaigle2019clustering}  projections that map functional data to scalars or low-dimensional vectors. Following \cite{garrett2024validating}, we propose to use deterministic spherical convolutional slices. Convolutional slicing integrates $f \in \mathcal{F_{\mathcal S\times T}}$ against a kernel function centered at location $s \in \mathcal{S}$, $k_s : \mathcal{S} \mapsto \mathbb{R}$ to generate a univariate function $f_s(t)$ as
\begin{equation}\label{eq:slice}
f_s(t) = \int_{\mathcal{S}} f(u,t) k_s(u; \theta) \, du,
\end{equation}
where $k_s$ is any spatially continuous function with a positive spectral density on $\mathcal{S}$ and $\theta \in \Theta$ denotes the parameters (bandwidth) of the kernel. We specify a Wendland function \citep{wendland1998error} to provide compactness, see Section \ref{sec:tb} for further discussion of our kernel choice.
Because this kernel is only supported in a neighborhood $B_\theta(s) \subset \mathcal{S}$ centered around $s \in \mathcal{S}$, the sliced function $f_s(t)$ represents a local composite of all nearby component functions of $f$. 
Given another space-time function $g(u, t) \in \mathcal{F}_{S \times T}$, we can integrate $g$ against the same kernel to generate 
\begin{equation}\label{eq:slice2}
g_s(t) = \int_{\mathcal{S}} g(u,t) k_s(u; \theta) \, du,
\end{equation}
which, analogously, is a local composite of the nearby component functions of $g$. 

Given the univariate sliced functions $f_s(t)$ and $g_s(t)$, we can immediately compute the elastic distance (\ref{d:elastic}) between them, which we call the local sliced elastic distance. We then define a global sliced elastic distance, denoted as the sliced elastic distance, as an average over all local sliced elastic distances, i.e., for all $s \in \mathcal{S}$.
To simplify notation, we will henceforth suppress the $(s,t)$ in $f(s,t)$ when there is no risk of confusion.

\begin{defn}[Sliced elastic distance]
\label{d:sed}
\label{d:fs}
Let $f,g\in\mathcal{F}_{\mathcal S\times \mathcal T}$. We define the sliced elastic distance, $D_{SE}$,  between $f$ and $g$ to be the vector consisting of three components: sliced amplitude distance $(D_{SA})$, sliced phase distance $(D_{SP})$, and sliced translation distance $(D_{ST})$:
\begin{equation}
  D_{SE}(f,g) 
  = \begin{bmatrix}D_{SA}(f,g) \\ D_{SP}(f,g) \\ D_{ST}(f,g)\end{bmatrix}
  = \begin{bmatrix}\Bigl\{\int_{\mathcal S}D_A(f_{s},g_{s})^2d{s}\Bigr\}^{1/2} \\ \Bigl\{\int_{\mathcal S}D_P(f_{s},g_{s})^2d{s}\Bigr\}^{1/2} \\ \Bigl\{\int_{\mathcal S}D_T(f_{s},g_{s})^2d{s}\Bigr\}^{1/2}\end{bmatrix},
\end{equation}
where $f_{s}$ and $g_{s}$ are the slices defined in (\ref{eq:slice}) and $D_A$, $D_P$, and $D_T$ are the amplitude, phase, and translation distances between two univariate functions as defined in (\ref{eq:3distances}). 
\end{defn} 

Since $D_{SA}$ is calculated via the amplitude distance between slices $f_s$ and $g_s$ at each location $s\in \mathcal S$, the time warping step is allowed to vary over space.
According to (\ref{eq:3distances}), the amplitude distance is based on the derivative of each slice through the SRVF representation. Thus, $D_{SA}$ integrates local differences between the dynamics of two climate fields in the vertical direction.
The phase distance between $f_s$ and $g_s$ represents the severity of the warping required to align the SRVFs of $f_{s}$ and $g_{s}$. In the context of comparing climate fields, $D_P(f_{s},g_{s})$ quantifies local differences in the timing of weather events or seasonal changes at location $s$, and $D_{SP}(f,g)$ measures the average time misalignment between the two climate fields. 
Sliced translation distance only captures differences between two functions at $t=0$. Though it is a necessary component to ensure that the sliced elastic distance satisfies the properties of a vector-valued metric, it is of little scientific interest for climate field comparison. In practice, we can substitute the sliced translation distance by a different measure of bias that is more of interest, such as bias in the annual mean or over a particular season of interest. We focus on $D_{SA}$ and $D_{SP}$ in our study.   

\begin{figure}[h]
\centering
\includegraphics[width=\linewidth]{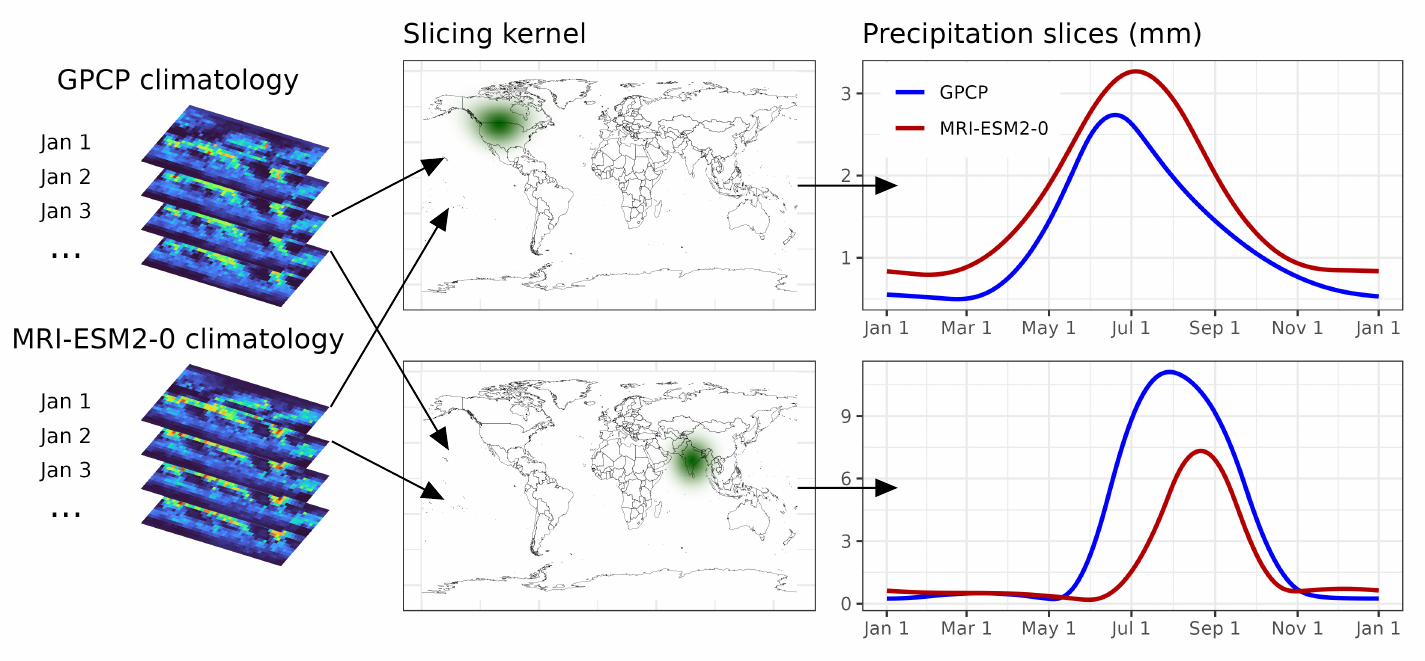}
\caption{An example of constructing the slice functions for GPCP and one CMIP6 model, MRI-ESM2-0. Each day, the spatial fields from each dataset are converted into many univariate functions, or slice functions, through a set of kernel projections. The two plots in the right panel show the slice functions at  two locations where the kernels in the central panel are centered.
We compute the elastic distance for each pair of slice functions using (\ref{eq:3distances}), and the sliced elastic distance between the two climate fields is obtained by averaging the elastic distances for each pair of slice functions.}
\label{f:diagram_sed}
\end{figure}

The sliced elastic distance allows us to compare high-dimensional spatiotemporal data through one-dimensional slices indexed by spatial locations. 
Each slice represents one perspective of the data and specifically captures the features of the data at a given location via the kernel convolution in (\ref{eq:slice}). 
Together, the slices provide a comprehensive view of the spatiotemporal field.
Figure \ref{f:diagram_sed} illustrates the concept of climate field comparison using the sliced elastic distance. 
The kernel convolution approach facilitates the comparison of climate model output and observational data that are available at different spatial resolutions.
The amplitude and phase distances for each slice are computed using the standard dynamic programming algorithm \citep{joshi2007novel}, thus addressing computational concerns associated with the time alignment step for multivariate functional data \citep{bernal2021computing, tucker2022dimensionality, hartman2021supervised}.  
Details on the computation of the sliced elastic distance are provided in Appendix \ref{a:computation}.

Climate data can also be viewed as multivariate functions of time, where the dimensionality of the multivariate function corresponds to the number of spatial locations. Under this framework, we could directly apply traditional EFDA methods to multivariate functional data \citep{joshi2007novel, srivastava2016functional}. 
However, this approach would be severely limited by the fact that warping function would not vary over space, and thus not allow us to characterize spatially varying timing biases  \citep{bernal2021computing, tucker2022dimensionality, hartman2021supervised}. Because climate data exhibit strongly spatially varying phase characteristics, implementing such an approach would result in the loss of much of the relevant information.
In contrast, our approach easily allows for spatial aggregation and spatially varying phase characteristics through the action of the kernel function.

\subsection{Theoretical properties}\label{sec:sedtheory}
We first show that the elastic distance $D_E(f, g)$ defined in (\ref{d:elastic}) is a valid vector-valued metric on $\mathcal F_\mathcal T$ and then establish our main result that the sliced elastic distance $D_{SE}(f, g)$ is a valid vector-valued metric on $\mathcal F_{\mathcal S\times \mathcal T}$.   
Both results rely on the idea of vector-valued metric spaces \citep{sastry2012common, rao2015vector, jachymski2016around},
which generalizes the concept of metric spaces to allow for multiple distance functions. 

\cite{jachymski2016around} introduced a simple characterization of vector-valued metrics which states that $(X,D)$ is a vector-valued metric space if and only if $D$ is a family of pseudometrics $(D_1,...,D_m)$ such that for any $x,y \in X$, $x\neq y \Rightarrow D_i(x,y) > 0$ for some $i\in{1,...,m}$. This characterization fits naturally with the EFDA representation as a vector of three pseudo-metrics: amplitude, phase, and translation.
Each is a valid metric on a different space, but not on the original functional data space, $\mathcal{F_T}$.

We use the characterization of vector-valued metrics from \cite{jachymski2016around} along with properties from \cite{srivastava2016functional} to prove that $D_E(f, g)$ is a vector-valued metric (Lemma \ref{l:elastic} stated in Appendix \ref{a:VVMSD}) as well as Theorem \ref{t:VVMSD}, a general result extending sliced vector-valued metrics to global vector-valued metrics.

\begin{thm}
\label{t:VVMSD}
If $D=(D_1,...,D_m)^T$ is a vector-valued metric on $\mathcal{F_T}$, and $f_s(t)$ and $g_s(t)$ are respectively the slice functions of $f(u,t)\in\mathcal{F_{S\times T}}$ and $g(u,t)\in\mathcal{F_{S\times T}}$ using a spatially continuous kernel $k(u;\theta)$ with a positive spectral density on spherical domain $\mathcal S\in \mathbb S^2$ as defined in (\ref{eq:slice}), then the vector-valued function $D_S = (D_{S1},...,D_{Sm})^T$ with  each component defined as \begin{equation*}
    D_{Si}(f,g) = \left\{\int_{\mathcal S}D_i\left(f_{s},g_{s}\right)^2 ds \right\}^{1/2}, \quad i=1,...,m,
\end{equation*}
is a vector-valued metric on $\mathcal F_{\mathcal S\times \mathcal T}$.
\end{thm}

The proof of Theorem \ref{t:VVMSD} is provided in Appendix \ref{a:VVMSD}. This theorem demonstrates that spherical convolutional slicing enables the properties of any metric or vector-valued metric on $\mathcal{F_T}$ to extend to the space $\mathcal{F_{S\times T}}$. 
Corollary \ref{c:SED} follows directly from Lemma \ref{l:elastic} and Theorem \ref{t:VVMSD} and establishes the sliced elastic distance (\ref{d:sed}) as a valid vector-valued metric.
 
\begin{cor}
\label{c:SED}
The sliced elastic distance $D_{SE}(f(s,t),g(s,t))$ is a vector-valued metric on $\mathcal{F_{S\times T}}$.
\end{cor}

Because the sliced elastic distance is a vector-valued metric on $\mathcal{F}_{\mathcal{S} \times \mathcal{T}}$, its components - the sliced amplitude, phase, and translation distances - provide a comprehensive suite of metrics for comparing spatiotemporal fields.
While our theoretical results are shown for spatial data with a spherical domain, similar results can be proven for any spatial domain on which a comparable convolution theorem to the spherical case \citep[introduced in][]{driscoll1994computing} is available.

\subsection{Timing bias and kernel function}
\label{sec:tb}
The sliced elastic distance between two functions, $f(u,t)$ and $g(u,t)$, is calculated as the global mean of the local elastic distances between the localized sliced functions $f_s$ and $g_s$, for $s \in \mathcal{S}$. 
The local elastic distances themselves can also serve as informative intermediate output of the procedure by indicating regions of severe misalignment.
As discussed in Appendix \ref{a:computation}, the code implementation for calculating $D_A$ and $D_P$ assumes that $\gamma^*_{f_s} = I(t)$ for identifiability, where $I(t) = t, t \in \mathcal{T}$ is the identity warping function \citep{tucker2013generative}. Under this assumption, no time warping is applied to $f_s$, so $\gamma^*_{g_s}(t)$ represents the optimal warping function that aligns $g_s$ with $f_s$. Therefore, for a given time $t \in \mathcal{T}$, we can calculate the timing bias, denoted as $B(f_s,g_s;t)$, of $g_s$ relative to $f_s$ as follows:
\begin{equation}
\label{eq:timing}
    B(f_s,g_s;t) = \gamma^*_{g_s}(t) - I(t).
\end{equation}
In contrast to phase distance, timing bias focuses on a single user-specified time point. Additionally, timing bias can be either positive or negative, representing a late or early timing of an event in $g_s$ compared to $f_s$, respectively. In our application, it is natural to treat the observed precipitation data as $f$ and the climate model data as $g$. Thus, $B(f_s, g_s; t)$ quantifies the timing biases of events of interest in the climate model relative to the observed data.

Valid choices of kernel functions include the Kent distribution function \citep{kent1982fisher} and the generalized Wendland functions \citep{wendland1998error}. Compact kernels are desirable for our application because they facilitate efficient computations and ensure each slice represents climate features in a relatively small local neighborhood, though ultimately our method provides a global characterization of the misalignment between fields. We thus choose the Wendland kernel function as in \cite{nychka2015multiresolution}:
\begin{equation}\label{eq:wendland}
    k_s(u;r) = 
    \begin{cases}
        \left(1-\frac{|s-u|}{r}\right)^6\left(35\frac{|s-u|^2}{r^2}+18\frac{|s-u|}{r}+3\right)/3 & |s-u| \leq r, \\
        0 & |s-u| > r,
    \end{cases}
\end{equation} 
where the range parameter $r$ determines the compactness of the kernel. As long as $r$ is smaller than or equal to the diameter of the Earth, the Wendland functions will be positive definite on $\mathbb S^2$ \citep{hubbert2023generalised}, satisfying the assumptions in equation (\ref{eq:slice}). In our application, $|s-u|$ is the chordal distance between locations $s$ and $u$, represented as latitude-longitude coordinates. Chordal distance accounts for the spherical geometry of our data and is chosen instead of the great circle distance for theoretical convenience as shown in Appendix \ref{a:VVMSD}.

\section{Simulation}
\label{sec:sim}

We conduct numerical experiments to evaluate the skill of our method in separating amplitude and phase differences in spatiotemporal fields, investigate the sensitivity of our method to the choice of range parameter, and compare our method to a traditional precipitation evaluation method.

\subsection{Settings}
To make our simulation realistic, we convert the GPCP precipitation observations over 1997-2014 to climatologies by computing the mean at each location on each calendar day of the year, excluding leap days. These GPCP climatologies serve as the baseline data $f(s,t)$. We then generate $g(s,t)$ by applying a series of spatially-varying amplitude and phase modifications on $f(s,t)$. 
Let $a_i(s)$, $i = 1,2,3$ and $p_j(s)$, $j = 1,2,3$ be the sets of parameters used for the amplitude and phase modifications, and let ${g}_{i,j}(s,t)$ denote the modified versions of GPCP corresponding to parameters $a_i(s)$ and $p_j(s)$. We obtain $g_{i,j}(s,t)$ by transforming $f(s,t)$ for all $s\in\mathbb{S}^2$ as follows 
\begin{equation}\label{eq:sim1}
g_{i,j}(s,t) = a_i(s) f\left(s,t^{p_j(s)}\right).
\end{equation}
The amplitude parameter $a_i(s)$ acts as a multiplier on the scale of $f(s,t)$. We restrict $a_i(s)$ to values greater than 1, such that larger values of $a_i(s)$ correspond to increased precipitation relative to the original $f(s,t)$, thereby increasing the amplitude distance. The phase parameter $p_j(s)$, which is restricted to positive values, introduces time warping to adjust the seasonal timing of $f(s,t)$ at each location. Values greater than 1 introduce late timing biases, while values between 0 and 1 introduce early timing biases.

The parameter fields $a_i(s)$ and $p_j(s)$ are both generated as functions of latitude.
The amplitude parameters, $a_i(s)$, all start at a value of 1.1 at the South Pole and then increase linearly at different rates toward the North Pole, reaching values of 1.15, 1.2, and 1.25 for $i = 1, 2, 3$, respectively.
This modification magnifies amplitude everywhere and intensifies with latitude. 
The phase parameters, $p_j(s)$, vary exponentially with latitude according to the functions $1.2^{\text{lat}(s)/90}$, $1.4^{\text{lat}(s)/90}$, and $1.6^{\text{lat}(s)/90}$ for $j=1,2,$ and $3$, respectively, where $\text{lat}(s)$ is the latitude coordinate of location $s$ in degrees.
In the southern hemisphere, the latitude coordinates are negative leading to $p_j(s)<1$, so early timing biases are introduced in that region. On the contrary, the modifications introduce late timing biases in the northern hemisphere.
For both $a_i(s)$ and $p_j(s)$, the most extreme amplitude and phase modifications occur when $i=j=3$ and the least extreme modifications occur when $i=j=1$.

\begin{figure}[h]
\centering
\begin{minipage}{0.7\linewidth}
  \begin{subfigure}{\linewidth}
  \centering
  \includegraphics[width=\linewidth]{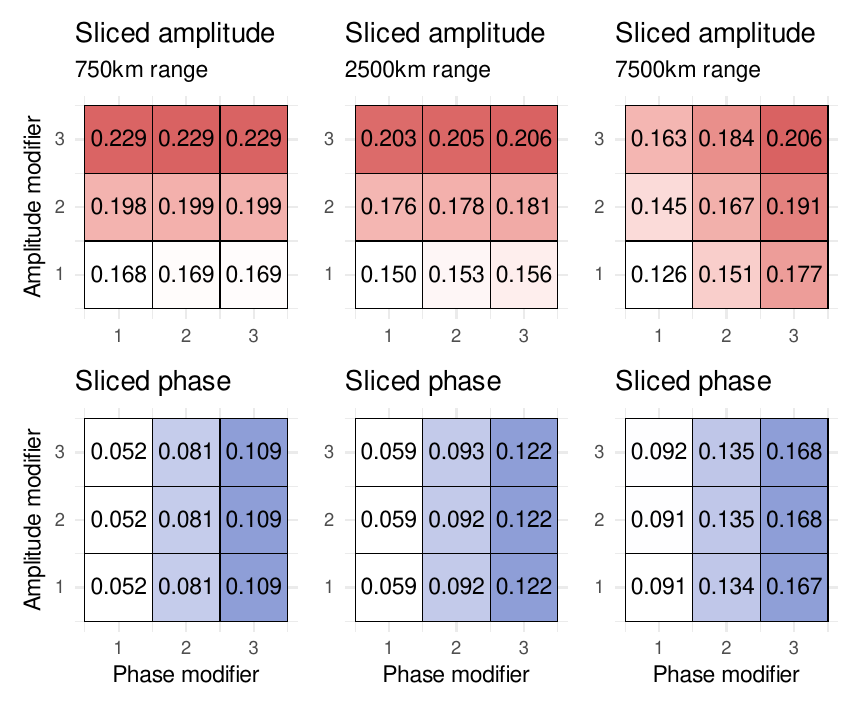}
  \caption{Sliced Elastic Distance Tables}
  \label{fig:sim1_range}
  \end{subfigure}
\end{minipage}
\hfill
\begin{minipage}{0.29\linewidth}
  \begin{subfigure}{\linewidth}
  \centering
  \includegraphics[width=\linewidth]{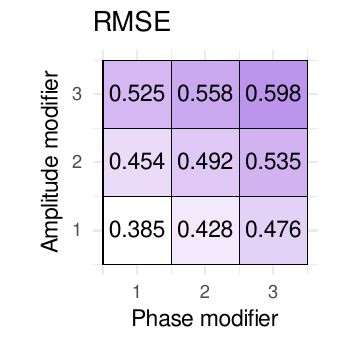}
  \caption{RMSE Table}
  \label{fig:sim1_rmse}
  \end{subfigure}
\end{minipage}
\vspace{2mm}
\caption{Panel (a): the sliced amplitude distance (red) and sliced phase distance (blue) between the original GPCP precipitation climatologies, $f(s,t)$, and the modified versions, $g_{i,j}(s,t)$, at various kernel range parameter values. Panel (b): the RMSE (purple) between the original and modified GPCP climatologies.  
For all plots, levels of the amplitude modifier $(i=1,2,3)$ and phase modifier $(j=1,2,3)$ are labeled in the $y$- and $x$-axes, respectively. 
Larger values of the amplitude and phase modifiers correspond to larger modifications made to $f(s,t)$. Color fill is determined independently per table, with lighter shades representing low distances and darker shades representing high distances.}
\label{fig:sim1_results}
\end{figure}

We calculate the sliced amplitude and sliced phase distance between $f(s,t)$ and each $g_{i,j}(s,t)$ following the algorithm described in Appendix \ref{a:computation}. To understand the influence of range parameter $r$ in the Wendland kernel function $k(u;r)$ on the elastic distances, we repeat the calculation for three different $r$ values in km: 750, 2500, and 7500. The first value, $r=750$ represents the value used in Section \ref{sec:results}. The remaining values represent potential choices for larger kernel ranges, but all values are less than the Earth's diameter (approximately 12,742km) to ensure positive definiteness of the kernel function.

\subsection{Results}
Simulation results are reported in Figure \ref{fig:sim1_results}. 
Overall, the sliced amplitude and sliced phase distance patterns show that our method is able to separate the spatially-varying amplitude differences from the spatially-varying phase differences. However, the range parameter values can affect this ability. 
For the range value of 750km, there is little to no influence of the phase modifier $p_j(s)$ on the sliced amplitude distance, evidenced by the very consistent amplitude distances at the three different phase modifier values. Vice versa, we also see no influence of the amplitude parameter $a_i(s)$ on the sliced phase distance values. 
Similar patterns are observed even if the range value increases to 2500km.
Whereas, for the range parameter of 7500km, the sliced amplitude distances increase with $p_j(s)$, showing evidence of entangled sliced amplitude and sliced phase distances.
This is because data that are further away in space have more distinct phase characteristics in our simulation.

Ignoring phase variability when taking functional means is known to distort the underlying structure \citep{tucker2013generative}. So, when using larger range parameter values, the kernel convolution used to create the slices acts as a cross-sectional weighted mean of misaligned functional data, leading to the entangled distances.  To more accurately quantify the amplitude and phase distances, we recommend choosing smaller range parameter values to decrease the influence of phase variability within the kernel radius. However, the grid size of data products must be considered in the choice of range parameter. If the range parameter is chosen to be significantly smaller than the grid size of the data, the slice functions may not contain the spatial information at the desired level.

Figure \ref{fig:sim1_results} (b) shows the RMSE calculated as the Euclidean norm between the original and modified climatologies. Unlike the sliced elastic distance, RMSE is unable to distinguish amplitude variability from phase variability, so RMSE provides distances that are essentially a joint reflection of the amplitude and phase modifications. 

\begin{figure}[h]
\centering
\includegraphics[width=1\linewidth]{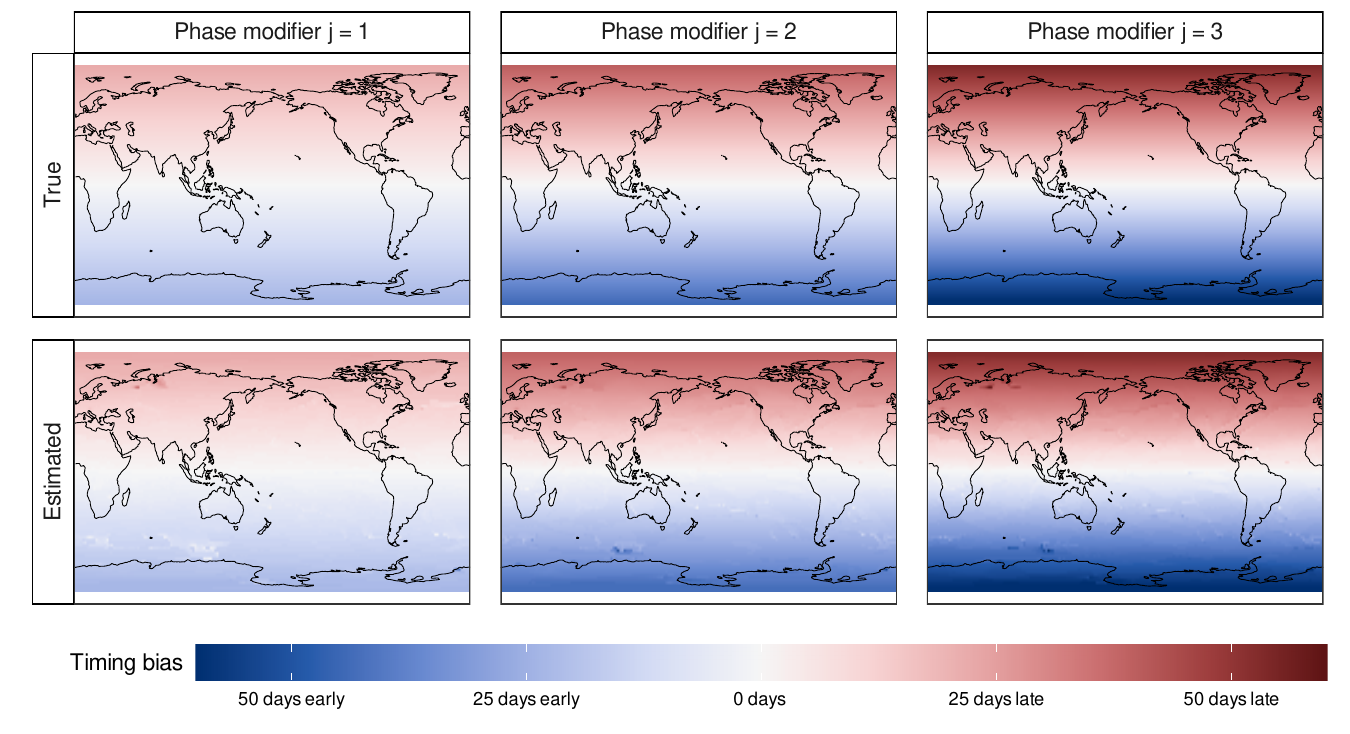}
\caption{Maps of timing bias on July 2 in the simulated data example for the three levels of the phase modifier $p_j(s), j=1,2,3$. The top row shows the true timing biases at each location in the modified GPCP climatologies. The bottom row shows the timing biases estimated using our sliced elastic distance method.
Slices are computed using the $750$km kernel, and timing biases are calculated using the formula in (\ref{eq:timing}). The color fill represents the magnitude and direction of timing biases.}
\label{fig:sim1_tb}
\end{figure}

In addition to quantifying the overall phase distance between two spatiotemporal fields, our method can also provide detailed information about timing biases at specific time points. We evaluate these intermediate results by comparing the estimated timing biases using the formula in (\ref{eq:timing}) to the true values in the modified GPCP climatologies incurred by the modifications in (\ref{eq:sim1}).
The top row of Figure \ref{fig:sim1_tb} shows maps of the true timing biases on July 2 (the midpoint of the calendar year), while the bottom row shows maps of the estimated values on the same day using the 750km kernel. These maps offer both the value and direction (early/late) of timing biases at each location for all three levels of phase modifications. 
In all cases, the fields are near-identical between the true and estimated values. The correct magnitude and direction of timing biases are recovered, with late biases near the north pole, early biases near the south pole, and no bias near the equator.

\section{Climate Model Evaluation}
\label{sec:results}

We first evaluate the skill of the CMIP6 climate models in reconstructing the historical climatologies on a global scale. 
We compute the sliced elastic distance between the daily precipitation fields of each CMIP6 model output and the GPCP observations. 
We also repeat this process for the ERA5 and NCEP datasets against GPCP to evaluate the performance of reanalysis fields, thus providing a baseline for climate model evaluation. 
We then demonstrate how to use the intermediate results from the sliced elastic distance calculation to quantify the timing bias for the onset and retreat of the Indian Summer Monsoon. We begin by illustrating the procedure at a single location, and then present the timing bias map across the entire monsoon region for a cohort of six CMIP6 models. 

Before computing the distance, raw precipitation values from each dataset are converted into daily climatology fields by taking the mean over 1997-2014 at each location on each calendar day of the year, excluding leap days. Although climate model evaluation typically focuses on monthly rather than daily data to eliminate ``weather" and retain only the ``climate", the temporal resolution of monthly data is often too coarse to accurately capture phase variability and timing biases. Therefore, we opt to use daily data, but instead of using the raw data directly, we will utilize a smoothed version of the daily data.
Quadratic trend filtering \citep{tibshirani2011solution} is designed to estimate the underlying continuous trend from noisy data. We apply this method to the daily climatology at each location to obtain a smooth and continuous function. This approach reduces the noise in the daily observations, making the resulting functional data comparable in smoothness to monthly data,  while preserving the temporal frequency of the daily data. 

For all sliced elastic distance computations, we follow the steps outlined in Appendix \ref{a:computation} and use the Wendland kernel function in (\ref{eq:wendland}) with a range of $r = 750$ km.
This range is large enough to cover more than a three-grid cell area near the equator in the lowest-resolution CMIP6 models (BCC-ESM1 and CanESM5, $128\times 64$ longitude-latitude resolution), ensuring that some degree of spatial smoothing is applied to all models. 
On the other hand, this choice is small enough to avoid issues with large range parameter values where sliced amplitude and sliced phase distances become entangled as demonstrated in Section \ref{sec:sim}.

\subsection{Global evaluation of CMIP6 precipitation models}

We create a 45-member ensemble of historical CMIP6 model outputs for daily total precipitation from January 1997 through December 2014, and compute the sliced elastic distance between the climatologies for each model and the GPCP data during this time period.
A smaller sliced phase distance indicates closer agreement in the timing of events and seasons, while a smaller sliced amplitude distance indicates that the model more closely matches the observations after phase alignment. As mentioned earlier, we also compute the distance from the ERA5 and NCEP reanalysis fields to GPCP. 

\begin{figure}[h]
\begin{center}
\includegraphics[width=1.04\linewidth]{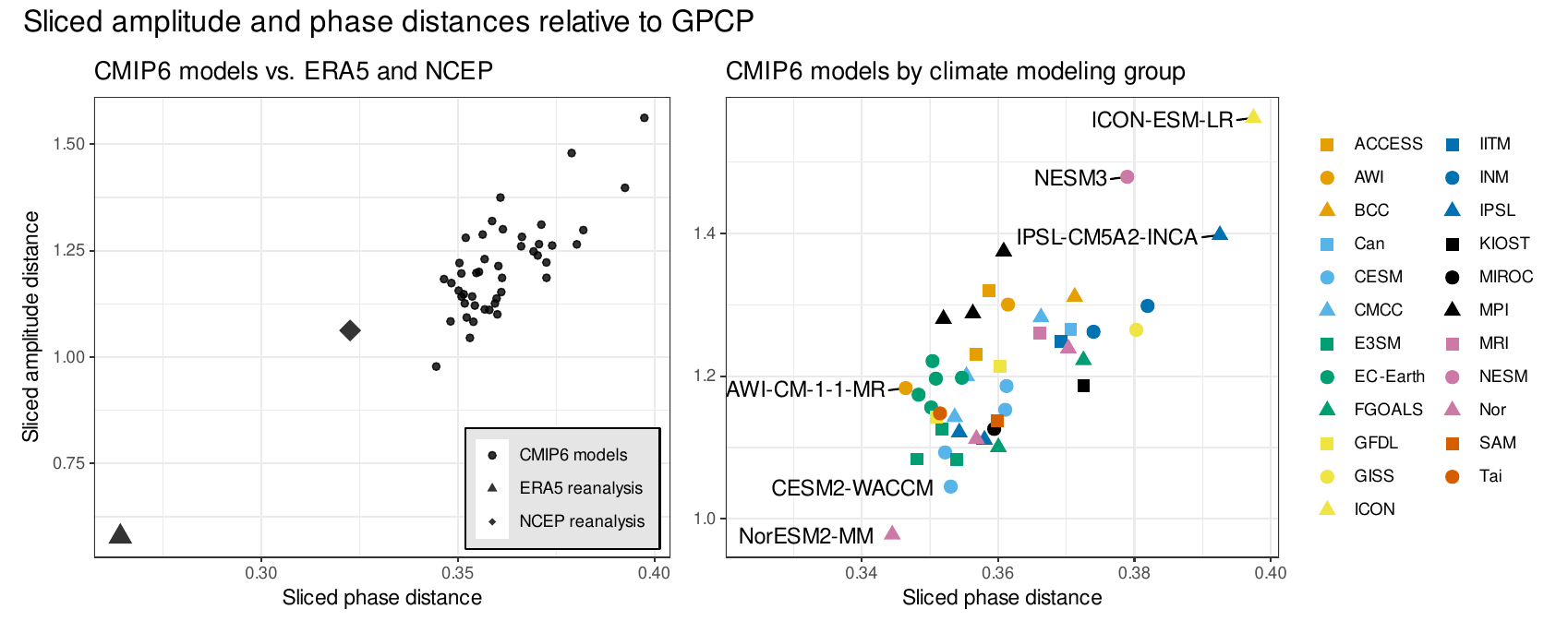}
\end{center}

\caption{Sliced elastic distance from the CMIP6 model outputs and ERA5/NCEP reanalyses to the GPCP observations. The left plot shows the CMIP6 models alongside ERA5 and NCEP. The right plot zooms in on the CMIP6 models only. In both plots, each model/reanalysis dataset is displayed as a point with the $x$ and $y$ axes representing the sliced phase and amplitude distances to GPCP, respectively.}\label{fig:cmip6_rankings}
\end{figure}

Figure \ref{fig:cmip6_rankings} shows that ERA5 has substantially lower sliced amplitude and phase distances to GPCP than all CMIP6 models and NCEP, indicating a higher degree of similarity to the observed data. NCEP has a lower sliced phase distance to GPCP than all CMIP6 models. However, two models, NorESM2-MM and CESM2-WACCM, have a lower sliced amplitude distance than NCEP to GPCP. Despite known issues with reanalysis datasets for precipitation \citep{tapiador2017global}, our method shows greater agreement between reanalysis datasets and observations than model outputs and observations, particularly in sliced phase distance. 

Among the CMIP6 models, the scatter plot shows evidence of a positive correlation between sliced amplitude and sliced phase distance. This implies that a model which performs well in one component (amplitude or phase) likely performs well in the other.
The Norwegian Earth System Model (NorESM2-MM) from the Norwegian Climate Center ranks the best, with the lowest sliced amplitude distance by a wide margin and the second-lowest sliced phase distance. The Community Earth System Model (CESM) Whole Atmosphere Community Climate Model configuration (CESM2-WACCM) ranks second for sliced amplitude distance, while the Alfred Wegener Institute Climate Model (AWI-CM-1-1-MR) achieves the lowest sliced phase distance, albeit by a small margin. The Energy Exascale Earth System Models (E3SM) exhibit a favorable balance between sliced amplitude and phase distances.

In contrast, the Icosahedral Non-hydrostatic Earth System Model (ICON-ESM-LR) shows the highest values for both sliced amplitude and phase distances. The Nanjing University of Information Science and Technology Earth System Model (NESM3) and the interactive aerosols/atmospheric chemistry configuration of the Institut Pierre-Simon Laplace coupled climate model (IPSL-CM5A2-INCA) have the second-highest sliced amplitude and sliced phase distances, respectively. Notably, for some modeling groups, such as E3SM and EC-Earth, the models within each group exhibit strong similarities in their sliced amplitude and phase distance values.

For  full details on our sliced elastic distance rankings (including sliced translation distance), refer to Table \ref{tab:full} in Appendix \ref{a:table}. This table also includes comparisons with two commonly used evaluation metrics for precipitation data: RMSE and MAE. Similar to the sliced amplitude and sliced phase distances, both RMSE and MAE rank ERA5 as the most similar to GPCP. However, both RMSE and MAE place NCEP among the average CMIP6 models. In most cases, the rankings for RMSE and MAE align closely with those of sliced amplitude distance, though there are some notable exceptions, such as the FGOALS models and ACCESS-CM2.
Sliced phase distance, however, provides a largely unique perspective on the rankings. Among all the different ranking methods, only sliced phase distance ranks NCEP above all the CMIP6 models. This is particularly significant, as time variability is an important feature in determining the similarity to observed data.

The sliced elastic distance provides a global assessment of a model's performance in mimicking the magnitude and timing of observed climatologies. However, it may also be desirable to identify where and when the differences between a model and observations occur.  Indeed, the intermediate results from calculating the sliced elastic distance can exactly reveal such information. We demonstrate this additional feature of our method in the context of Indian Summer Monsoon in the following section.

\subsection{Timing bias in the Indian Summer Monsoon region}
\label{sec:monsoon}

We characterize phase differences between GPCP and the CMIP6 model outputs by focusing on a key component of the global climate system: the Indian Summer Monsoon (ISM). Specifically, we examine the Monsoon Core Region (MCR), defined as the area of India from $15^\circ$N to $30^\circ$N latitude and $68^\circ$E to $88^\circ$E longitude. The MCR encompasses regions with the highest proportion of rainfall during the monsoon season, which is typically considered to span June, July, August, and September (JJAS), compared to the remaining eight months of the year.

\begin{figure}[h]
\centering
\includegraphics[width=\linewidth]{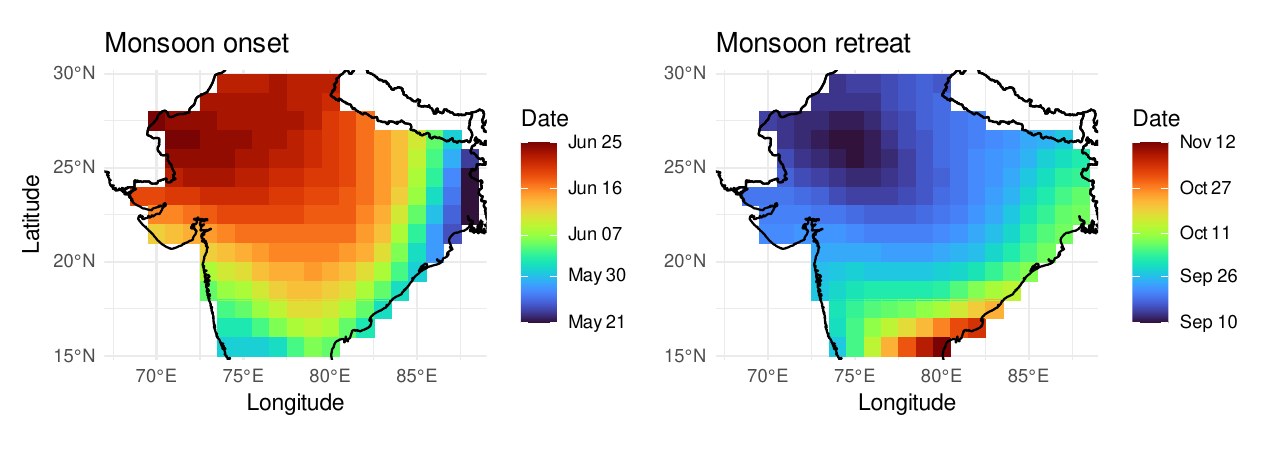}
\caption{Maps of ISM onset and retreat dates in the GPCP climatology slices.
The dates are determined for each location in the MCR using a 50\% threshold of maximum precipitation. Color fill is used to indicate the onset/retreat date.}
\label{fig:gpcp_dates}
\end{figure}

Given the profound socio-economic importance of the onset and retreat of the summer monsoon, we focus on quantifying the timing biases for these key events across the MCR. Various methods have been proposed for determining the onset and retreat of the monsoon season at each location based on climatological precipitation or other aspects of the hydrological cycle \citep{wang2002rainy, fasullo2003hydrological, misra2018local}. In this study, we adopt a similar approach, defining the onset and retreat based on thresholds of the maximum climatological precipitation at each location in the GPCP dataset. Specifically, the onset date at a location is the first day in the climatology when the precipitation exceeds 50\% of the maximum daily rainfall, while the retreat date is the last day when the precipitation exceeds 50\% of the maximum daily rainfall.

Figure \ref{fig:gpcp_dates} shows the onset and retreat dates for each slice location in the MCR. These maps are not intended to exactly reproduce previous results, but rather to provide a per-slice definition of the onset and retreat, which will serve as the reference for our phase analysis.
Compared to the maps of onset and retreat dates in \cite{misra2018local}, our results are overall very similar, but exhibit more spatial smoothness due to the kernel convolution in our approach and the more coarse precipitation observations in our data.

\begin{figure}[h]
\centering
\includegraphics[width=\linewidth]{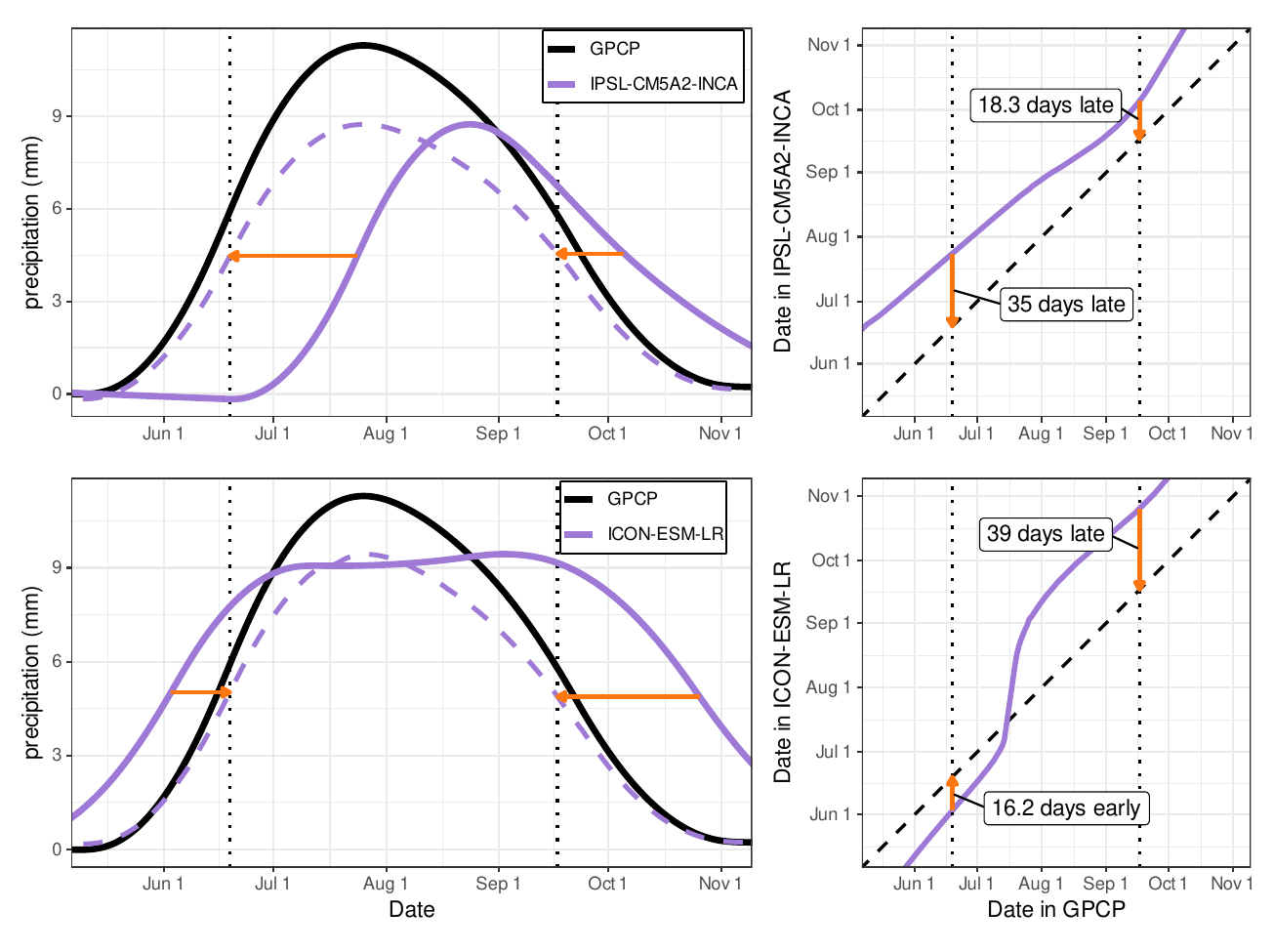}
\caption{
ISM onset/retreat timing biases between IPSL-CM5A2-INCA/ICON-ESM-LR and GPCP for the slice located at $22.5^\circ$N, $78^\circ$E. Orange arrows indicate the time warping applied to each model at the ISM onset (June 19th) and retreat (September 17th), which are displayed as dotted lines. Left-side plots show the model (purple) and GPCP (black) precipitation slices. Post-alignment model slices are shown as dashed lines. 
Right-side plots display the time warping functions for each model in purple against the identity function (dashed line).
}
\label{fig:efda}
\end{figure}

For a given location in the MCR, we can take the onset and retreat dates from the maps in Figure \ref{fig:gpcp_dates} and then compute the timing bias, defined in (\ref{eq:timing}), from the CMIP6 models to GPCP on those dates. Figure \ref{fig:efda} demonstrates the calculation procedure for IPSL-CM5A2-INCA and ICON-ESM-LR, the two models with the highest global sliced phase distance. This demonstration focuses on a single location in the middle of the MCR ($22.5^o$N, $78^o$E).
For IPSL-CM5A2-INCA, \cite{sepulchre2020ipsl} previously established a one-month lag in the ISM event. Their work focuses on monthly data while ours uses smoothed daily data, allowing us to determine timing biases at a finer temporal resolution. At this particular location, we find that IPSL-CM5A2-INCA exhibits a late onset bias of about 35 days and a late retreat bias of about 18 days compared to GPCP, so our results generally agree with \cite{sepulchre2020ipsl}. In contrast, ICON-ESM-LR exhibits an early onset bias of about 16 days and a late retreat bias of about 39 days.  

To understand the spatially-varying timing biases in the onset and retreat of the Indian Summer Monsoon, we repeat the procedure in Figure \ref{fig:efda} for every location in the MCR using a cohort of six CMIP6 models. These models were selected from those with the highest and lowest sliced elastic distances in Figure \ref{fig:cmip6_rankings}. The results are displayed as maps in Figure \ref{fig:tb_maps}. For the maps of onset timing bias, we observe that NorESM2-MM, AWI-CM-1-1-MR, and CESM2-WACCM exhibit little to no timing variability in the MCR. Among these, AWI-CM-1-1-MR shows timing biases closest to zero across the entire MCR. In contrast, ICON-ESM-LR, IPSL-CM5A2-INCA, and NESM3 show significant biases in some or all parts of the MCR. Notably, IPSL-CM5A2-INCA displays a strong positive timing bias of four or more weeks throughout the region.

For the maps of retreat timing bias, NorESM2-MM, AWI-CM-1-1-MR, and CESM2-WACCM all exhibit greater levels of timing bias compared to the onset event. For NorESM2-MM and AWI-CM-1-1-MR, there is a mix of small timing biases both in terms of early and late retreat, while CESM2-WACCM features mostly early or neutral timing biases. The remaining models exhibit mainly late retreat biases. For ICON-ESM-LR, late retreat biases occur in almost all areas of the MCR, while for IPSL-CM5A2-INCA and NESM3, late retreat biases occur most prominently at the western edge and center of the MCR. The large biases around the western edge in some models may be related to the representation of topography in those models.

\begin{figure}[h!]
\centering
\begin{subfigure}{\linewidth}
  \includegraphics[width=\linewidth]{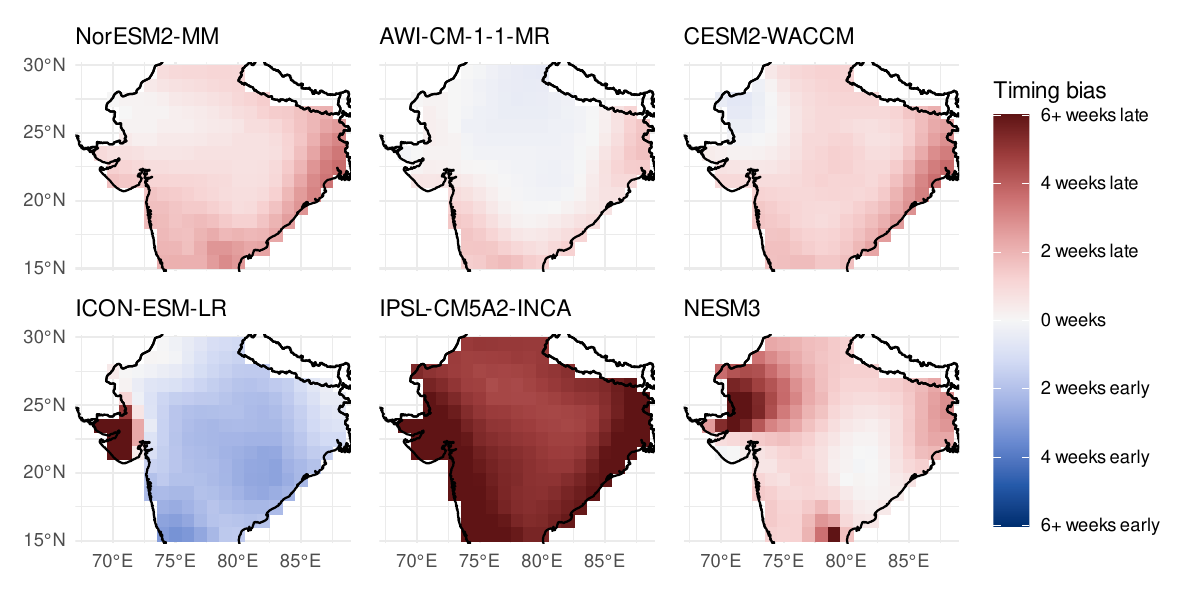}
  \caption{Maps of onset timing bias}
  \label{fig:high}
\end{subfigure}
\begin{subfigure}{\linewidth}
  \centering
  \includegraphics[width=\linewidth]{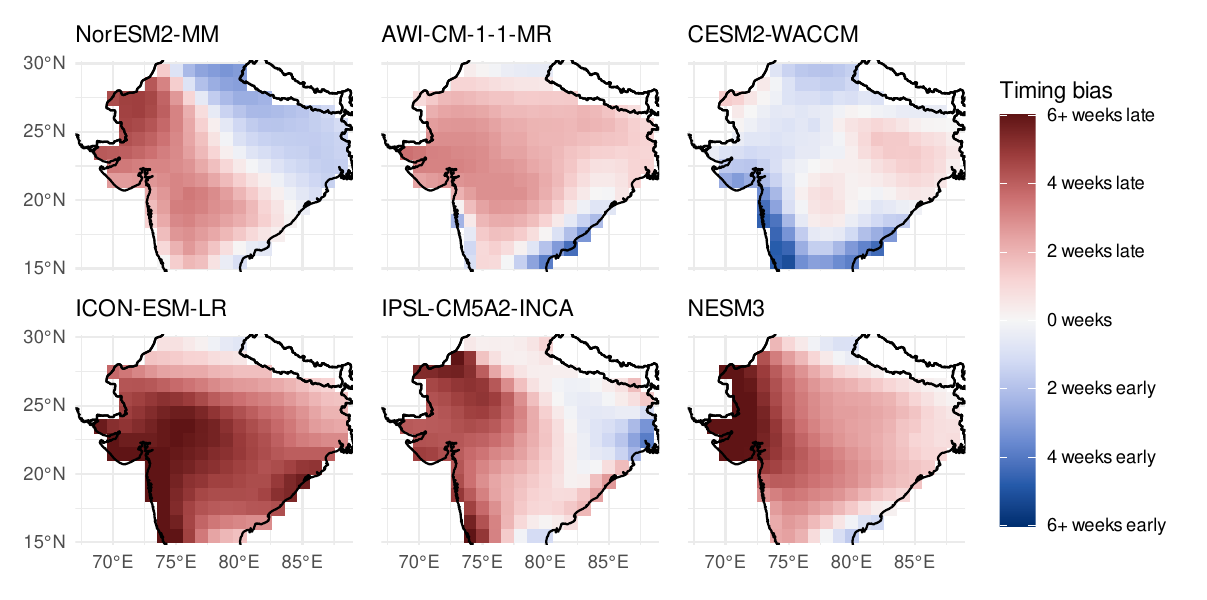}
  \caption{Maps of retreat timing bias}
  \label{fig:low}
\end{subfigure}
\vspace{-2.5mm}
\caption{Timing bias in the ISM onset and retreat dates for six CMIP6 models. Timing biases are computed at each location following the procedure in Figure \ref{fig:efda}. At each location, timing bias is indicated with color fill. Panel (a) and (b) display the onset and retreat biases, respectively.}
\label{fig:tb_maps}
\end{figure}

Overall, our results clearly characterize timing biases in the onset and retreat of the Indian Summer Monsoon, providing climate modelers with a useful diagnostic tool to understand a model's performance in phase variability. We particularly notice large regions of late arrival and retreat biases in many models. Such information is not available when one simply uses summary statistics such as the mean JJAS rainfall to evaluate the skill of climate model in simulating monsoon. Furthermore, the timing biases shown in Figure \ref{fig:tb_maps}, particularly the late retreat biases, indicate that some models simulate a significant amount of monsoon precipitation outside of the JJAS time interval. Therefore, a model with low JJAS rainfall may not actually underestimate the ISM precipitation, but instead mistime the event.

\section{Discussion}
\label{sec:discussion}

To aid climate model evaluation for precipitation, we developed a new metric called sliced elastic distance which quantifies model performance by separately accounting for spatial and temporal variability. The sliced elastic distance is a vector-valued metric that captures  the amplitude, phase, and translation distance between two spatiotemporal processes. It extends the traditional elastic distance by incorporating spatially varying time warping through a kernel convolution procedure. We focus on amplitude and phase to provide a unique evaluation of precipitation models. Phase variability corresponds to errors in the timing of a system evolution, such as seasonal transitions of monsoons, while amplitude variability reflects differences in the underlying precipitation patterns after phase alignment. Distinguishing these two components offers deeper insight into the sources of model misfit. By applying our method to evaluate CMIP6 precipitation models, we quantified how each model performs relative to observed data based on their amplitude and phase distances. In addition, we examined the performance of CMIP6 models in capturing the proper onset and retreat timing of the Indian Summer Monsoon at a local spatial scale.

One limitation of our approach is that the slicing process does not include a spatial alignment in addition to the standard time alignment.
Spatial warping was used in \cite{levy2014correcting} to correct errors in the location of precipitation events in climate models, but adapting this approach for our daily analysis would be computationally infeasible. 
Another limitation is that our method assumes a closed and bounded time domain, i.e. $[0, 1]$ arbitrarily defined by setting January 1st to 0 and December 31st to 1. This does not allow for inter-year warping, as an open time domain, i.e. $\mathbb{S}^1$ \citep{sebastian2003aligning}, would, which may potentially distort the intra-year alignments.
While such closed temporal domains are commonly used in elastic FDA \citep{joshi2007novel}, only the properties of the amplitude (shape) distance have been established on this domain \citep{srivastava2016functional}. 
Lastly, while our threshold-based method for determining local onset and retreat dates is convenient for shape analysis, alternative approaches exist for defining monsoon onset and retreat, which may yield different results \citep[i.e.,][]{misra2018local}. Timing biases based on those dates could produce slightly different results, but the conclusion is likely qualitatively robust.

Our method was developed in the context of evaluating precipitation models. However, our method can be applied to compare any two spatiotemporal fields where phase is of interest or matters. 
For example, the sliced elastic distance can also be used to assess climate models for temperatures (including the evolution of ENSO events), or for climate model tuning \citep{hourdin2017art}.
Beyond climate science, potential applications of our method include comparing spectral data from remote sensing products, analyzing functional measurements from atomic-scale microscopy in materials science, or analyzing $360^\circ$ videos, where the spherical data considerations in our method are directly applicable.

\section{Acknowledgments}

We acknowledge the World Climate Research Programme, which, through its Working Group on Coupled Modelling, coordinated and promoted CMIP6. We thank the climate modeling groups for producing and making available their model output, the Earth System Grid Federation (ESGF) for archiving the data and providing access, and the multiple funding agencies who support CMIP6 and ESGF.
NCEP/DOE Reanalysis II data provided by the NOAA PSL, Boulder, Colorado, USA, from their website at \url{https://psl.noaa.gov}.

Robert Garrett was partially supported by NSF-DGE-1922758 from the National Science Foundation. Bo Li was partially supported by NSF-DMS-2124576 from the National Science Foundation. We would like to thank Robert Krueger for the interesting discussions on convolution theorems and spherical harmonics.

\setlength{\bibsep}{4pt}
\bibliographystyle{rss}
\bibliography{ref}

\newpage 
\Appendix
\setcounter{page}{1}
\appendix

\section{Spatial resolution of CMIP6 model outputs}
\label{a:data1}

\begin{figure}[h]
\centering
\label{a:data}
\resizebox{\columnwidth}{!}{%
    \begin{tabular}{lcc}
    \toprule
    \textbf{Model} & \textbf{Longs} & \textbf{Lats} \\
    \midrule
    ACCESS-CM2 & 192 & 144 \\
    ACCESS-ESM1-5 & 192 & 145 \\
    AWI-CM-1-1-MR & 384 & 192 \\
    AWI-ESM-1-1-LR & 192 & 96 \\
    BCC-ESM1 & 128 & 64 \\
    CESM2 & 288 & 192 \\
    CESM2-FV2 & 144 & 96 \\
    CESM2-WACCM & 288 & 192 \\
    CESM2-WACCM-FV2 & 144 & 96 \\
    CMCC-CM2-HR4 & 288 & 192 \\
    CMCC-CM2-SR5 & 288 & 192 \\
    CMCC-ESM2 & 288 & 192 \\
    CanESM5 & 128 & 64 \\
    E3SM-1-0 & 360 & 180 \\
    E3SM-2-0 & 360 & 180 \\
    E3SM-2-0-NARRM & 360 & 180 \\
    EC-Earth3 & 512 & 256 \\
    EC-Earth3-AerChem & 512 & 256 \\
    EC-Earth3-CC & 512 & 256 \\
    EC-Earth3-Veg & 512 & 256 \\
    EC-Earth3-Veg-LR & 320 & 160 \\
    FGOALS-f3-L & 288 & 180 \\
    FGOALS-g3 & 180 & 80 \\
    \midrule
    \bottomrule
  \end{tabular}
\quad
\begin{tabular}{lcc}
\toprule
\textbf{Model} & \textbf{Longs} & \textbf{Lats} \\
\midrule
    GFDL-CM4 & 288 & 180 \\
    GFDL-ESM4 & 288 & 180 \\
    GISS-E2-2-G & 144 & 90 \\
    ICON-ESM-LR* & N/A & N/A \\
    IITM-ESM & 192 & 94 \\
    INM-CM4-8 & 180 & 120 \\
    INM-CM5-0 & 180 & 120 \\
    IPSL-CM5A2-INCA & 96 & 96 \\
    IPSL-CM6A-LR & 144 & 143 \\
    KACE-1-0-G & 192 & 144 \\
    KIOST-ESM & 192 & 96 \\
    MIROC6 & 256 & 128 \\
    MPI-ESM-1-2-HAM & 192 & 96 \\
    MPI-ESM1-2-HR & 384 & 192 \\
    MPI-ESM1-2-LR & 192 & 96 \\
    MRI-ESM2-0 & 320 & 160 \\
    NESM3 & 192 & 96 \\
    NorCPM1 & 144 & 96 \\
    NorESM2-LM & 144 & 96 \\
    NorESM2-MM & 288 & 192 \\
    SAM0-UNICON & 288 & 192 \\
    TaiESM1 & 288 & 192 \\
    \midrule
    \bottomrule
    \\
\end{tabular}}
\vspace{2mm}
\captionof{table}{List of obtained CMIP6 models and their spatial resolutions. Note that ICON-ESM-LR was obtained on an icosahedral grid made up of 10,242 cells.}
\end{figure}

\section{Proof of Theorem \ref*{t:VVMSD}}
\label{a:VVMSD}

Proof of Theorem \ref{t:VVMSD} relies on the convolution theorem for finite-dimensional unit spheres \citep{driscoll1994computing}, as well as Lemmas \ref{l:slice} and \ref{l:pseudoSD}. Lemma \ref{l:slice} asserts that the slicing operation in equation (\ref{eq:slice}) produces valid functional data in $\mathcal{F_T}$, and Lemma \ref{l:pseudoSD} states that pseudometrics on $\mathcal{F_T}$ extend to pseudometrics on $\mathcal{F}_{\mathcal{S} \times \mathcal{T}}$ when our slicing method is applied. We first develop Lemma \ref{l:elastic}, then technical lemmata \ref{l:slice} and \ref{l:pseudoSD}, to finally prove Theorem \ref{t:VVMSD}.

\begin{lem}
\label{l:elastic}
$D_{E}(f,g)$ is a vector-valued metric on $\mathcal{F_T}$.
\end{lem}

\begin{proof}
Proposition 2.1 in \cite{jachymski2016around} provides a clear structure for our proof. First, we show that amplitude, phase, and translation distance are each a pseudometric on $\mathcal{F_T}$. Then, we finish the proof by showing the final property: for any $f,g \in \mathcal{F_T}$, $f\neq g \Rightarrow D_A(f,g)>0,$ $D_P(f,g)>0,$ or $D_T(f,g)>0$.

Section 4.10.1  in \cite{srivastava2016functional} shows that amplitude distance is a proper metric on the quotient space $\mathbb{L}_2/\Gamma$. Since for all functions $f,g \in \mathcal{F_T}$, there exist unique orbits $[q_f],[q_g]\in\mathbb{L}_2/\Gamma$, the identity, symmetry, and triangle inequality properties hold trivially on $\mathcal{F_T}$ as well.

Section 4.10.2 in \cite{srivastava2016functional} shows that for two functions $f,g\in\mathcal{F_T}$, phase distance satisfies the Identity, Symmetry, and Triangle Inequality properties on $\mathcal{F_T}$. So, phase distance is a pseudometric on $\mathcal{F_T}$.

For translation distance, let $f,g,h \in \mathcal{F_T}$. Note $f(0),g(0),h(0)\in\mathbb{R}$. Then,
\begin{enumerate}
    \item $D_T(f,f) = |f(0)-f(0)| = 0$
    \item $D_T(f,g) = |f(0)-g(0)| = |g(0)-f(0)|= D_T(g,f)$
    \item $D_T(f,h) = |f(0)-h(0)| \leq |f(0)-g(0)|+|g(0)-h(0))| = D_T(f,g)+D_T(g,h)$.
\end{enumerate}
Therefore, translation distance is a pseudometric on $\mathcal{F_T}$.

It remains to show that for all $f,g \in \mathcal{F_T}$, $f\neq g \Rightarrow D_A(f,g)>0,$ $D_P(f,g)>0,$ or $D_T(f,g)>0$. We proceed by proving the contrapositive: for all $f,g\in\mathcal{F_T}$, if $D_A(f,g)=D_P(f,g)=D_T(f,g)=0$ then $f=g$. Let $f,g\in\mathcal{F_T}$ with $D_A(f,g)=D_P(f,g)=D_T(f,g)=0$. Denote the relative phase of $f$ with respect to $g$ as $(\gamma_f^*,\gamma_g^*)$. $D_P(f,g) = 0 \Rightarrow \gamma_f^* = \gamma_g^*$. So,
\begin{align}
0 &= D_A(f,g) \\
  &= \mathrm{inf}_{\gamma_f,\gamma_g\in\tilde\Gamma_I} ||(q_{f},\gamma_f) - (q_{g},\gamma_g)|| \\
  &= ||(q_{f},\gamma_f^*) - (q_{g},\gamma_g^*)|| \\
  &= ||(q_{f},\gamma_f^*) - (q_{g},\gamma_f^*)|| \\
  &= ||q_{f} - q_{g}||.
\end{align}
The third line follows by the definition of relative phase. The last line follows by Lemma 4.2 in \cite{srivastava2016functional}.

It follows by metricity of the Euclidean norm that $q_f = q_g$. In addition, $D_T(f,g) = 0 \Rightarrow f(0)=g(0)$. So, we have $f(0)=g(0)$ and $q_f = q_g$. By the absolute continuity of $f$ and $g$, it follows that $f=g$. So, for all $f,g\in\mathcal{F_T}$, $D_A(f,g)=D_P(f,g)=D_T(f,g)=0 \Rightarrow f=g$.

We have shown that the elastic distance $D_E(f,g)$ is a family of pseudometrics on $\mathcal{F_T}$ that jointly satisfies the positivity property on $\mathcal{F_T}$. So, by Proposition 2.1 in \cite{jachymski2016around}, it follows that $D_E$ is a vector-valued metric on $\mathcal{F_T}$.
\end{proof}

\begin{lem}
\label{l:slice}
If $f \in \mathcal{F_{S\times T}}$, then the slice $f_{s}(t)$ as defined in equation (\ref{eq:slice}) is an element of $\mathcal{F_T}$ for all $s\in \mathcal{S}$.
\end{lem}

\begin{proof}
Let $f\in \mathcal{F_{S\times T}}$, $s^*\in \mathcal{S}$, and $r>0$. Suppose for contradiction that $f_{s^*}(t)\notin \mathcal{F_T}$. Then, $f_{s^*}(t)$ is not absolutely continuous, so there exists some $\epsilon^*>0$ such that for all $\delta >0 $, there exists a finite sequence of $M$ sub-intervals $(a_m,b_m)\in[0,1]$, $m=1,...,M$, with $\sum_{m=1}^M (b_m-a_m)<\delta$ and $\sum_{m=1}^M|f_{s^*}(a_m)-f_{s^*}(b_m)|\geq\epsilon^*$. 

Let $m_{max} = \text{arg max}_{m\in{1,...,M}}|f_{s^*}(a_m)-f_{s^*}(b_m)|$. It follows that $|f_{s^*}(a_{m_{max}})-f_{s^*}(b_{m_{max}})| \geq \epsilon^*/M$. For clarity, denote $t_1 = a_{m_{max}}$ and $t_2 = b_{m_{max}}$. Note that $t_2-t_1 < \delta$. We have:

\begin{align}
\epsilon^*/M &\leq |f_{s^*}(t_1)-f_{s^*}(t_2)| \\
  &=\left|\int_{\mathcal{S}}f(s,t_1)k_{s^*}(s;r)ds-\int_{\mathcal{S}}f(s,t_2)k_{s^*}(s;r)ds\right| \\
  &=\left|\int_{\mathcal{S}}\left\{f(s,t_1)-f(s,t_2)\right\}k_{s^*}(s;r)ds\right| \\
  &\leq\left[\int_{\mathcal{S}}\left\{f(s,t_1)-f(s,t_2)\right\}^2ds\right]^{1/2}\left\{\int_{\mathcal{S}}k_{s^*}(s;r)^2ds\right\}^{1/2},
\end{align}
where the last line follows by the Cauchy-Schwarz Inequality. This implies:
\begin{align}
\frac{(\epsilon^*/M)^2}{\int_{\mathcal{S}}k_{s^*}(s;r)^2ds} &\leq\int_{\mathcal{S}}\left\{f(s,t_1)-f(s,t_2)\right\}^2ds \\
&\leq \text{max}_{s\in \mathcal{S}}\left\{f(s,t_1)-f(s,t_2)\right\}^2.
\end{align}
Now, taking the square root of each side, we get:
\begin{equation}
\frac{\epsilon^*/M}{\left\{\int_{\mathcal{S}}k_{s^*}(s;r)^2ds\right\}^{1/2}} 
\leq \text{max}_{s\in \mathcal{S}}|f(s,t_1)-f(s,t_2)|. \\
\end{equation}
Denote the quantity on the left side of the above inequality as:
\begin{equation}
\epsilon = \frac{\epsilon^*/M}{\left\{\int_{S^2}k_{s^*}(s;r)^2ds\right\}^{1/2}} .
\end{equation}
Since $k_{s^*}(s;r) \geq 0$ for all $s,s^*\in \mathcal{S}$, $r>0$ implies that $\int_{\mathcal{S}}k_{s^*}(s;r)^2ds>0$. So, $\epsilon>0$. Let $s' =  \text{arg max}_{s\in \mathcal{S}}|f(s,t_2)-f(s,t_1)|$. Now, we have found $\epsilon>0$ such that for all $\delta>0$, there exists a sub-interval $(t_1,t_2)\subset[0,1]$ with $t_2-t_1<\delta$ and:
\begin{equation}
|f(s',t_2)-f(s',t_1)| \geq \epsilon.
\end{equation}
We have found a location $s'\in \mathcal{S}$ such that $f(s',t)$ is not an absolutely continuous function of time, so we have a contradiction of the assumption that $f \in \mathcal{F_{S\times T}}$. Our proof by contradiction is complete, therefore $f_{s}(t)\in \mathcal{F_T}$ for all $s\in \mathcal{S}$ and $r>0$.
\end{proof}

\begin{lem}
\label{l:pseudoSD}
If $D$ is a pseudometric on $\mathcal{F}$ then $D_S$, the sliced version of $D$, is a pseudometric on $\mathcal{F_{S\times T}}$.
\end{lem}
\begin{proof}
Let $D$ be a pseudometric on $\mathcal{F_T}$. Let $D_S$ be the sliced distance function corresponding to $D$ following the process in Theorem \ref{t:VVMSD}. We begin by proving the identity property. Let $f(s,t)\in\mathcal{F_{S\times T}}$. Then,
\begin{align}
    D_S(f(s,t),f(s,t)) &= \left\{\int_{S^2}D(f_s(t),f_s(t))^2ds\right\}^{1/2}\\
        &=\left\{\int_{S^2}0^2ds\right\}^{1/2} \\
        &= 0.
\end{align}
The second line holds by the identity property of $D$. Next, we show the symmetry property. Let $f(s,t),g(s,t)\in\mathcal{F_{S\times T}}$. Then,
\begin{align}
    D_S(f(s,t),g(s,t)) &= \left\{\int_{S^2}D(f_s(t),g_s(t))^2 ds\right\}^{1/2} \\
                         &= \left\{\int_{S^2}D(g_s(t),f_s(t))^2 ds\right\}^{1/2} \\
                         &= D_S(g(s,t),f(s,t))
\end{align}
The second line holds by the symmetry property of $D$. Finally, we show the triangle inequality property. Let $f(s,t),g(s,t),h(s,t) \in \mathcal{F_{S\times T}}$. Then,
\begin{align}
    D_S(f(s,t),h(s,t)) &= \left\{\int_{S^2}D(f_s(t),h_s(t))^2 ds\right\}^{1/2} \\
                         &\leq \left[ \int_{S^2} \Bigl\{ D(f_s(t),g_s(t)) + D(g_s(t),h_s(t)) \Bigr\}^2 ds \right]^{1/2}\\
                         &\leq \left\{\int_{S^2}D(f_s(t),g_s(t))^2ds\right\}^{1/2}+\left\{\int_{S^2}D(g_s(t),h_s(t))^2ds\right\}^{1/2} \\
                         &= D_S(f(s,t),g(s,t))+D_S(g(s,t),h(s,t))
\end{align}
The second line holds by the triangle inequality property of $D$ on $\mathcal{F_T}$. The third line holds by the Minkowski Inequality.

Finally, we have that for a pseudometric $D$ on $\mathcal{F_T}$, the sliced version of $D$, $D_S$, satisfies all three properties of a pseudometric on $\mathcal{F_{S\times T}}$. Therefore, if $D$ is a pseudometric on $\mathcal{F_T}$, then $D_S$ is a pseudometric on $\mathcal{F_{S\times T}}$.
\end{proof}

Using the previous lemmas along with results from \cite{driscoll1994computing}, we can prove the previously stated Theorem \ref{t:VVMSD}.

\begin{repthm*}
If $D=(D_1,...,D_m)^T$ is a vector-valued metric on $\mathcal{F_T}$, and $f_s(t)$ and $g_s(t)$ are respectively the slice functions of $f(u,t)\in\mathcal{F_{S\times T}}$ and $g(u,t)\in\mathcal{F_{S\times T}}$ using a spatially continuous kernel $k(u;\theta)$ with a positive spectral density on spherical domain $\mathcal S\in \mathbb S^2$ as defined in (\ref{eq:slice}), then the vector-valued function $D_S = (D_{S1},...,D_{Sm})^T$ with  each component defined as \begin{equation*}
    D_{Si}(f,g) = \left\{\int_{\mathcal S}D_i\left(f_{s},g_{s}\right)^2 ds \right\}^{1/2}, \quad i=1,...,m,
\end{equation*}
is a vector-valued metric on $\mathcal F_{\mathcal S\times \mathcal T}$.
\end{repthm*}

\begin{proof} Let $D$ be a vector-valued metric on $\mathcal{F_T}$.
For each $i\in\{1,...,m\}$, let $D_{Si}$ be the sliced version of $D_i$ using kernel $k(u;\theta)$. By Lemma \ref{l:pseudoSD}, $D_S = (D_{S1},...,D_{Sm})$ is a family of pseudometrics on $\mathcal{F_{S\times T}}$. To show $D_S$ is a vector-valued metric on $\mathcal{F_{S\times T}}$, it suffices to show that for any $f,g \in \mathcal{F_{S\times T}}$, if $f\neq g$ then $D_{Si}(f,g) > 0$ for some $i\in\{1,...,m\}$. We proceed by proving the contrapositive: for all $f,g\in\mathcal{F_{S\times T}}$, if $D_S(f,g) = 0_m$ then $f=g$.

Let $f,g\in\mathcal{F_{S\times T}}$ with $D_S(f,g)=0_m$. Then $0 = \int_{\mathcal S}D_{Si}(f_{s},g_{s})ds$ for all $i\in\{1,...,n\}$. So, for all $i\in\{1,...,n\}$, $D_{Si}(f_{s},g_{s})=0$ for almost every $s\in \mathcal{S}$. By property of vector-valued metrics, this implies that $f_{s}(t) = g_{s}(t)$ for almost every $s\in \mathcal{S}$. Let $h(s,t)=f(s,t)-g(s,t)$. Note that $h_{s}(t) = \int_{\mathcal S}\{f(u,t)-g(u,t)\}k_{s}(u;\theta)du = f_{s}(t)-g_{s}(t)$. So, for all $t\in\mathcal T$, $h_{s}(t)=0$ for almost every $s \in \mathcal S$.

Now, fix $t\in\mathcal T$ and define the spatial convolution of $h$ with $k$ at time $t$ as $c_{h,t}(s) = \int_{\mathcal S} h(u,t)k_s(u;\theta)du$.
This convolution is a function of space only, serving as the spatial version of the previously defined slice functions, which are functions of time only.
Note that $c_{h,t}(s)$ is equal to $0$ for almost every $s\in \mathcal S$ because $h_{s}(t)=0$ for almost every $s \in \mathcal S$. Additionally, since $c_{h,t}(s)$ is defined as a convolution of continuous functions on $\mathcal S$, it is itself a continuous function on $\mathcal S$. It follows by property of continuity that $c_h(s,t)=0$ for all $s\in \mathcal S$. 

Using the spherical harmonics representation of $c_{h,t}(s)$ \citep{driscoll1994computing}, we can represent $c_{h,t}(s)$ as
\begin{equation}
  c_{h,t}(s) =\sum_{l\geq0}\sum_{|m|\leq l}\widetilde{c_{h,t}}(l,m)Y^m_l(s).
\end{equation}
Where $Y^m_l(s)$ are the spherical harmonics bases and $\widetilde{c_{h,t}}(l,m)$ are the spherical harmonics coefficients for $c_{h,t}(s)$. Since the bases $Y^m_l(s)$ are orthonormal and $c_{h,t}(s)=0$, we have that $\widetilde{c_{h,t}}(l,m)=0$ for all $l$ and $m$.

Since $c_h(s,t) = \int_{\mathcal S} h(u,t)k_s(u;\theta)du$ is a convolution of functions on $\mathcal S$, where $\mathcal S$ was previously defined to be the 2-dimensional unit sphere, using Theorem 1 in \cite{driscoll1994computing} we can write the spherical harmonics coefficients $\widetilde{c_{h,t}}(l,m)$ in terms of the spherical harmonics coefficients for $h(u,t)$ and $k_s(u;\theta)$, denoted respectively as $\widetilde{h^t}(l,m)$ and $\tilde{k}(l,m)$:
\begin{equation}
\widetilde{c_{h,t}}(l,m) = \alpha(l)\widetilde{h^t}(l,m)\tilde{k}(l,0),
\end{equation}
where $\alpha(l)=2\pi\sqrt{\frac{4\pi}{2l+1}}$. Clearly $\alpha(l)>0$ for all $l\geq 0$. 
By our assumption that $k(u;\theta)$ has positive spectral density on $\mathcal S$, we know that $\tilde{k}(l,0)>0$ for all $l\geq 0$.
Therefore, since $\widetilde{c_{h,t}}(l,m)=0$, we must have $\widetilde{h^t}(l,m)=0$ for all $l$ and $m$. Using the spherical harmonics representation for $h(s,t)$, we can see:
\begin{align}
  h(s,t) &= \sum_{l\geq0}\sum_{|m|\leq l}\widetilde{h^t}(l,m)Y^m_l(s)\\
  &= \sum_{l\geq0}\sum_{|m|\leq l}0*Y^m_l(s)\\
  &=0.
\end{align}
So, $h(s,t)=0$ for all $s\in \mathcal S$. Since $t$ was fixed arbitrarily, we also have that $h(s,t)=0$ for all $t\in\mathcal T$. Therefore, $0=h(s,t)=f(s,t)-g(s,t)$, implying that $f = g$. 

We have shown that for all $f,g\in \mathcal{F_{S\times T}}$, if $D_S(f,g)=0_m$, then $f=g$. Our proof by contrapositive is complete, therefore for all $f,g \in \mathcal{F_{S\times T}}$, if $f\neq g$, then $D_{Si}(f,g) > 0$ for some $i\in\{1,...,n\}$. So, for any vector-valued metric $D$ on $\mathcal{F}$, $D_S$ is a family of pseudometrics that jointly satisfies the positivity property on $\mathcal{F_{S\times T}}$. Therefore, by Proposition 2.1 in \cite{jachymski2016around}, if $D$ is a vector-valued metric on $\mathcal{F_T}$ and $k(u;\theta)$ is a continuous spatial kernel with positive spectral density on $\mathcal{S}$, then $D_S$, the sliced version of $D$ defined using $k(u;\theta)$, is a vector-valued metric on $\mathcal{F_{S\times T}}$.
\end{proof}

\newpage
\section{Sliced Elastic Distance Implementation}
\label{a:computation}

Before computing the sliced elastic distance, we estimate continuous functional data from the precipitation climatology at each location by applying quadratic trend filtering via the \verb|glmgen| R package using a smoothing parameter of $\lambda=1,250$ \citep{tibshirani2014adaptive}.
Because climate model output and observational data are discrete in both the space and time dimensions, the exact integrals in Definition \ref{d:sed} cannot be computed. We approximate these integrals with summations and averages over a discrete set of locations. 
It is the user's decision to choose locations at which they want to have the slice functions. This choice may depend on the grid density of data products and preferences for the spatial resolution of slices. In our simulation and data analysis, we choose a regular latitude-longitude grid, $G$, resolved by 180 latitude values and 360 longitude values, to represent the spatial domain. The resolution of $G$ is chosen to match the resolution of the GPCP data, but users can make their own choice for the grid size and structure. 
Given the two daily precipitation fields represented by $f(u,t)\in\mathcal{F_{S\times T}}$ and $g(u, t)\in\mathcal{F_{S\times T}}$ and the common grid $G$, the sliced elastic distance between $f$ and $g$ can be easily obtained through the following major steps: 

\label{algo:sed}
\begin{enumerate}
   \item {\bf Slicing} Compute the slice functions $f_s$ and $g_s$ for each location $s \in G$ through multiplication of $f(u,t)$ and $g(u,t)$ with the kernel $k_s(u)$ at each time point.
   \item {\bf Local elastic distances} At each location $s\in G$, compute approximate amplitude and phase distances $\tilde D_A (f_s, g_s)$ and $\tilde D_P (f_s, g_s)$ using the dynamic programming algorithm provided in R package \verb|fdasrvf| \citep{fdasrvf}. 
   \item {\bf Spatial weighting} Assign weights for each location $s \in G$, denoted as $w_s$, as the cosine of its latitude to adjust for the different areas of each grid cell. This follows the standard practice for global climate data, e.g., \cite{li2016comparison}. 
   \item {\bf Sliced elastic distance} Compute the approximate sliced elastic distance $\tilde D_{SE}$ between $f$ and $g$ as follows: 
   \begin{equation}
    \tilde D_{SE}(f,g) \approx 
    \begin{bmatrix}
    \Bigl\{W^{-1}\sum_{s\in G}w_s \tilde D_A(f_{s},g_{s})^2\Bigr\}^{1/2} \\ 
    \Bigl\{W^{-1}\sum_{s\in G}w_s \tilde D_P(f_{s},g_{s})^2\Bigr\}^{1/2} \\ 
    \Bigl\{W^{-1}\sum_{s \in G}w_s D_T(f_{s},g_{s})^2\Bigr\}^{1/2}
    \end{bmatrix},
\end{equation}
where $W^{-1}= {1}/\sum_{s\in G} w_s$. 
\end{enumerate}
See the supplemental material for our full code implementation in \verb|R|.
To produce the spatial maps, region means, and timing biases discussed in Section \ref{sec:method}, intermediate values for $\tilde D_A(f_s,g_s), \tilde D_P(f_s,g_s), D_T(f_s,g_s)$ and the relative phase functions $\gamma^*_{f_s},\gamma^*_{g_s}$ are saved at each location $s\in G$.
Note that when calculating the elastic distance between two functions $f,g\in\mathcal{F_T}$, the \verb|fdasrvf| implementation assumes that $\gamma^*_f(t) = I(t) = t$, the identity warping function, for identifiability. 

\section{Table of Global Rankings}
\label{a:table}

\begin{figure}[H]
\centering
\includegraphics[width=\linewidth]{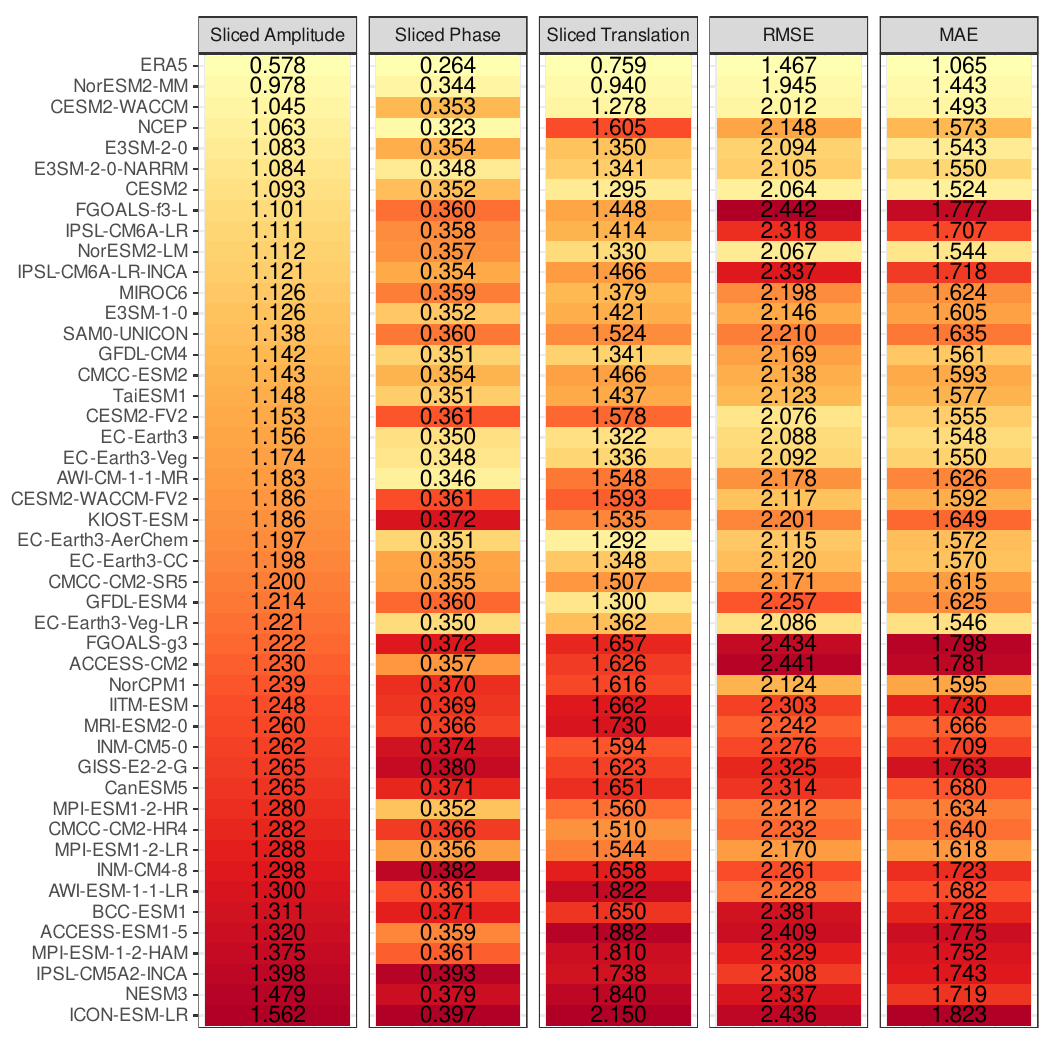}
\captionof{table}{CMIP6 daily precipitation models ranked based on similarity to GPCP. The ERA5 and NCEP Reanalyses are included with the CMIP6 models as a baseline. Distances are calculated between the climatologies using using sliced amplitude, phase, and translation distance (750km range) as well as RMSE and MAE. Color fill is used for visual comparison of the rankings from each distance, with yellow representing a low rank and red representing a high rank.}
\label{tab:full}
\end{figure}

\end{document}